\documentclass[11pt, a4paper]{article}
\usepackage[margin=2.5cm]{geometry}
\usepackage{amssymb}
\usepackage{amsmath}
\usepackage{amsthm}
\usepackage{float}
\usepackage{mathtools}
\usepackage{listings}
\usepackage{multirow}
\usepackage[dvipdfmx]{graphicx}
\usepackage{xcolor}
\usepackage{caption}
\usepackage{setspace}
\theoremstyle{definition}

\newtheorem{theorem}{Theorem}[section]
\newtheorem{proposition}[theorem]{Proposition}
\newtheorem{lemma}[theorem]{Lemma}

\newtheorem{corollary}[theorem]{Corollary}
\newtheorem{example}{Example}

\DeclareMathOperator*{\argmin}{arg\,min}

\newcommand{\dd}{\mathrm{d}}
\newcommand{\ee}{\mathrm{e}}
\newcommand{\mm}{\mathrm{m}}
\newcommand{\iter}[2]{{#1}^{(#2)}}
\newcommand{\thetait}[1]{\iter{\theta}{#1}}
\newcommand{\etait}[1]{\iter{\eta}{#1}}

\begin{document}
\begin{center}
{Minimization of Functions on Dually Flat Spaces \\
Using Geodesic Descent Based on Dual Connections} \\

\vspace{0.5cm}
{\large
Gaku Omiya$^{1}$  and\
Fumiyasu Komaki$^{1,2}$
}

\vspace{1em}
$^1$ Department of Mathematical Informatics, The University of Tokyo, Tokyo, Japan \\
$^2$
RIKEN Center for Brain Science (CBS), Wako, Saitama, Japan
\end{center}
\begin{abstract}
We propose geodesic-based optimization methods on dually flat spaces,
where the geometric structure of the parameter manifold is closely related
to the form of the objective function.
A primary application is maximum likelihood estimation in statistical models,
especially exponential families, whose model manifolds are dually flat.
We show that an m-geodesic update, which directly optimizes the log-likelihood,
can theoretically reach the maximum likelihood estimator in a single step.
In contrast, an e-geodesic update has a practical advantage in cases
where the parameter space is geodesically complete,
allowing optimization without explicitly handling parameter constraints.
We establish the theoretical properties of the proposed methods
and validate their effectiveness through numerical experiments.
\end{abstract}


\section{Introduction}
In optimization problems on Riemannian manifolds, we consider the following problem
\begin{equation}
    \label{optmfd}
\min_{x \in M} f(x).
\end{equation}
based on a Riemannian manifold $(M,g)$ and a smooth function $f$ defined on $M$.
For instance, by setting 
$M = S^{n-1}$, the $(n-1)$-dimensional unit sphere in $\mathbb{R}^n$,
and equipping $S^{n-1}$ with the induced metric from the Euclidean metric,
$S^{n-1}$ becomes a Riemannian manifold.
Furthermore, if we take $A$ as a symmetric matrix and define $f(x)\coloneqq x^\top Ax$, then problem (\ref{optmfd}) corresponds to finding the smallest eigenvalue of $A$.
In optimization problems on Riemannian manifolds of this kind,
optimization methods utilizing geodesics have been widely employed. 
These methods primarily rely on geodesics based on the Levi-Civita connection. 
More specifically, using the geodesic $\gamma^{(0)}$ determined by the Levi-Civita connection with initial conditions $\gamma^{(0)}(0)=\iter{x}{k},\dot{\gamma^{(0)}}(0) = - \operatorname{grad} f(\iter{x}{k})$, the update formula $x_{k+1} = \gamma^{(0)}(t_k)$ is commonly used \cite{smith1994}.

On the other hand, dually flat spaces $(M,g,\nabla,\nabla^{*})$, 
studied in information geometry and including important examples such as exponential families of probability distributions, 
admit two important dual connections, $\nabla$ and $\nabla^{*}$, which generally differ from the Levi-Civita connection.
The geodesics based on these dual connections are natural in the context of dually flat spaces and offer computational advantages. 
However, despite these benefits, optimization methods utilizing geodesics based on dual connections have not been extensively studied.

In this paper, we study optimization problems on dually flat spaces and investigate the properties that hold when updating in a manner analogous to the geodesic update based on the Levi-Civita connection, but using the geodesics 
$\gamma$ and $\gamma^*$ associated with the two dual connections, $\nabla$ and $\nabla^*$.
In particular, when the objective function is closely related to the geometric structure of the model manifold, such as likelihood or Kullback–Leibler divergence, our method is expected to be especially meaningful.

In Section 2, we review the notation and basic concepts for the subsequent sections.
In Section 3, we present results that do not depend on a specific objective function and highlight the relationship between our approach and existing methods.
In Section 4, we examine commonly used objective functions in statistics, such as the Kullback–Leibler divergence and the likelihood function, and discuss the properties of our optimization method for such objective functions.
In Section 5, we present numerical experiments for multiple examples.

\section{Geometric and Statistical Preliminaries}
In this section, we introduce the geometric and statistical concepts and assumptions 
used throughout the paper. 
For more details on differential geometry and information geometry, 
see Amari and Nagaoka \cite{amari2000}.

Throughout this paper, we consider $C^{\infty}$ manifolds. 
Unless otherwise specified, we assume that the dimension of a manifold $M$ is $m$.
Given a local coordinate system
$\xi =(\xi^1,\dots,\xi^m)$ on $M$, the differential map is defined as
$(\mathrm{d} \xi)_p \colon T_p M \to \mathbb{R}^m,
\sum_{i=1}^m c^i \bigl(\frac{\partial}{\partial \xi^i}\bigr)_p
\mapsto (c^1,\dots,c^m)^\top$.

For a Riemannian manifold $(M, g)$ and a smooth function $f \in C^\infty(M)$, 
the Riemannian gradient $\operatorname{grad} f(x) \in T_x M$ at a point $x \in M$ 
is defined in local coordinates by
\begin{equation}
\operatorname{grad} f
= \sum_{i=1}^m
\sum_{j=1}^m g^{ij}
\frac{\partial f}{\partial \xi^i}\, 
\frac{\partial}{\partial \xi^j}.
\end{equation}

The Christoffel symbols $\Gamma^{~~k}_{ij}$ of an affine connection $\nabla$ with respect to the coordinate system $\xi = (\xi^1, \dots, \xi^m)$ are defined by  
\begin{equation} \label{Christoffel}
\nabla_{\frac{\partial}{\partial \xi^i}}
\frac{\partial}{\partial \xi^j}
= \sum_{k=1}^m \Gamma^{~~k}_{ij} \frac{\partial}{\partial \xi^k}.
\end{equation}
If there exists a coordinate system $\xi$ with which the connection coefficients
$\Gamma_{ij}^{\; k}$ vanish on an open neighborhood, then the manifold is said
to be flat with respect to $\nabla$.
We consider several different affine connections.

A geodesic on $M$ is a smooth curve $\gamma \colon I \to M$ 
(with $I \subset \mathbb{R}$ an interval) that satisfies
\begin{equation}
 \label{eq:geod1}
 \nabla_{\dot{\gamma}} \dot{\gamma} = 0,
\end{equation}
where $\nabla_{\dot{\gamma}}$ denotes the covariant derivative along $\gamma$.  
For a coordinate expression $\xi(\gamma(t)) = (\xi^1(t), \dots, \xi^m(t))$,
\eqref{eq:geod1} is represented as
\begin{equation}
 \frac{\dd^2 \xi^i}{\dd t^2}(t) + \sum_{j=1}^m \sum_{k=1}^m (\Gamma^{~~i}_{jk} \circ \gamma)(t)\, 
 \frac{\dd \xi^j}{\dd t}(t)\, 
 \frac{\dd \xi^k}{\dd t}(t) = 0
 ~~~~~ (i = 1, \dots, m).
\end{equation}

In information geometry, manifolds $(M,g,\nabla,\nabla^{*})$ equipped with dual affine connections are considered. 
If $(M,g,\nabla,\nabla^{*})$ is flat with respect to $\nabla$, then it is also flat with respect to $\nabla^{*}$. 
Such a structure is called a dually flat space. 
For any $p \in M$, there exists a neighborhood $U$ with an affine coordinate system $\theta = (\theta^1,\ldots,\theta^m)$
satisfying $\Gamma_{ij}^{~~k} = 0$
for $\nabla$, together with a convex potential function $\psi(\theta)$ such that
\[
 \eta_i = \frac{\partial}{\partial \theta^i}\psi(\theta), \quad i=1,\ldots,m,
\]
defines affine coordinates $\eta = (\eta_1,\ldots,\eta_m)$ for $\nabla^{*}$. 
To properly represent duality, we denote $\theta^i$ with upper indices and $\eta_i$ with lower indices.
The function $\psi(\theta)$ is called the potential function of the dually flat space. 
The components of the Riemannian metric in $\theta$-coordinates are
\[
 g_{ij} = \frac{\partial^2}{\partial \theta^i \partial \theta^j}\psi(\theta),
\]
and the components in $\eta$-coordinates form the inverse matrix $(g^{ij}) = (g_{ij})^{-1}$.

The Bregman divergence from a point $r$ to a point $q$ in $U$
is defined using
the potential function $\psi$ with respect to an affine coordinate system
$\theta$ corresponding to the $\nabla$-connection that is dual to $g$, and the
potential function $\varphi$ with respect to $\eta$, as
\[
B_\phi(r, q) \coloneqq \psi(\theta(r)) + \varphi(\eta(q))
- \sum_{j=1}^m \theta^j(r) \eta_j(q),
\]
where
\[
\varphi(\eta) \coloneqq \sum_i \theta^i \eta_i
- \psi(\theta(\eta)).
\]
This is also called the canonical divergence of a dually flat space.

Exponential families provide an important example. 
Equipped with the Fisher metric $g$, the e-connection $\nabla^{(\ee)}$, and the $m$-connection $\nabla^{(\mm)}$, an exponential family $(S,g,\nabla^{(\ee)},\nabla^{(\mm)})$ constitutes a dually flat space. 
Its density function can be written, with respect to a base measure, as
\[
 \exp \Biggl\{ \sum_{i=1}^m \theta^i T_i(x) - \psi(\theta) \Biggr\},
\]
where $T(x) = (T_1(x),\ldots,T_m(x))$ is the sufficient statistic. 
Here, $\theta$ is the natural parameter
which is an affine coordinate for $\nabla^{(\ee)}$,
$\eta_i(\theta) \coloneqq \mathrm{E}_\theta[T_i(x)]$
is the expectation parameter, and $\psi(\theta)$ the potential function of the dually flat structure.
Given an observation $x = (x_1, \ldots, x_m)$, the maximum likelihood estimator of $\eta$ is given by $\hat{\eta}_i = x_i$ for $i = 1, \ldots, m$.
The Kullback--Leibler divergence $\mathrm{KL}(r,q)$
from $r$ to $q$ coincides with the Bregman divergence $B_\psi(r,q)$.

In statistics, most exponential families used in practice belong to the class of regular exponential families, see \cite{BN1978} p.~116.  
Every distribution in a regular exponential family $S$ can be represented by its natural parameter $\theta$, where the parameter space $\theta(S) \subset \mathbb{R}^m$ is open and convex.  
Therefore, when applying the results in the following sections, it is usually sufficient to consider a single coordinate system that covers the entire family $S$, without the need for multiple local coordinate charts.
Moreover, when $\theta(S) = \mathbb{R}^m$, one can perform descent along e-geodesics without having to consider any constraints on the parameter space. For example, it is known that this holds when the family of distributions has finite support.

\vspace{0.5cm}
\noindent
Exapmle.

A random variable $X$ taking one of the values $1,\ldots,m+1$ with probabilities $\eta_i>0$ $(i=1,\ldots,m+1)$ is said to follow a categorical distribution, i.e.\ the one-trial $(m+1)$-nomial distribution. The family of categorical distributions is
\begin{align}
\label{multi}
S_{m+1} \coloneqq \Bigl\{ (\eta_1,\ldots,\eta_m)
\, \Big| \, \eta_i > 0 ~ (i=1,\ldots,m), ~
\eta_{m+1} \coloneqq 1 - \sum_{i=1}^m \eta_i > 0
\Bigr\}.
\end{align}
We refer to this family as the categorical distribution model. The unknown parameter is $\eta = (\eta_1,\ldots,\eta_m)$.

Given $N$ independent observations $x_1,\ldots,x_N$ from a categorical distribution, the log-likelihood is
\begin{align}
\label{eq:likelihood3}
\sum_{l=1}^N \log p(x_l \mid \eta) &= \sum_{l=1}^N \log \eta_{x_l}
=  \sum_{i=1}^{m+1} N T_i \log \eta_i
=  \sum_{i=1}^m N T_i \log \frac{\eta_i}{\eta_{m+1}}
+ N \log \eta_{m+1} \notag \\
&= N \Biggl\{ \sum_{i=1}^m T_i \log \theta^i
- \log \biggl( \sum_{i=1}^m \exp \theta^i + 1 \biggr) \Biggr\},
\end{align}
where $T_i \coloneqq \sum_{l=1}^N \delta_i(x_l)/N$ $(i=1,\ldots,m)$, $\delta_i(x_l)$ denotes the indicator function that equals $1$ when $x_l=i$ and $0$ otherwise, and $\theta^i = \log (\eta_i/\eta_{m+1})$. The potential function is $\psi(\theta) = \log \bigl( \sum_{i=1}^m \exp \theta^i + 1 \bigr)$.

The vector $(NT_1,\ldots,NT_{m+1})$ follows the $(m+1)$-nomial distribution with $N$ trials, and the maximum likelihood estimator of $\eta_i$ is given by $T_i$ $(i=1,\ldots,m)$.
When $T_i \geq 1$ for all $i \in \{1,\ldots,m+1 \}$,
the distribution $p(x \mid \hat{\eta})$
obtained by plugging in the maximum likelihood estimate belongs to the model $S_{m+1}$.
If $T_i = 0$ for some $i$, the distribution obtained by plugging in the maximum likelihood estimate lies on the boundary of the model and does not belong to the model itself.
\qed

\section{Optimization on Dually Flat Spaces}

In this section, we consider the minimization of a function $f$ on a dually flat space $(M, g, \nabla, \nabla^{*})$.
While many existing studies on function minimization on Riemannian manifolds employ descent methods along geodesics
based on the Levi-Civita connection associated with the Riemannian metric,
here we instead investigate optimization methods that perform descent along geodesics induced by the dual connection.
We analyze such descent methods on dually flat spaces and derive several of their fundamental properties.

We introduce an update formula based on dual connections.

\vspace{0.5cm}
\begin{lemma}\label{双対平坦空間における測地線}
Let $f \colon U \to \mathbb{R}$ be a $C^{\infty}$ function, and let $\gamma$ be
a $\nabla$-geodesic satisfying $\gamma(0)=\iter{p}{k}$ and
$\dot{\gamma}(0)=-\operatorname{grad}f(\iter{p}{k})$.
Assume that $\iter{t}{k}>0$ is chosen so that
$\gamma(s) \in U$ for all $s \in [0,\iter{t}{k}]$, and define
$\iter{p}{k+1} \coloneqq \gamma(\iter{t}{k})$.
Then we have
\begin{equation}
 \label{egeo2}
 \theta(\iter{p}{k+1})
 = \theta(\iter{p}{k})
 - t_k\, D_{\eta} f(\iter{p}{k}),
\end{equation}
where
$D_\eta f(\iter{p}{k}) \coloneqq
\left(
\frac{\partial f}{\partial \eta_1}(\iter{p}{k}),
\dots,
\frac{\partial f}{\partial \eta_m}(\iter{p}{k})
\right)$.
\hfill \qed
\end{lemma}
\vspace{0.5cm}

This result shows that the computation of a geodesic reduces to computing
partial derivatives with respect to an affine coordinate system, once such a
coordinate system is available.
Since affine coordinate systems are often explicitly obtainable for commonly
used exponential families, \eqref{egeo2} is useful in practice.

\begin{proof}[Proof of Lemma~\ref{双対平坦空間における測地線}]
Since $\theta$ is an affine coordinate system with respect to $\nabla$, 
for the $\nabla$-geodesic $\gamma$ satisfying $\gamma(0)=p$ and 
$\dot{\gamma}(0)=-\operatorname{grad} f(p)$, we have
\[
\theta(\gamma(s)) 
= \theta(p) - s \, (\mathrm{d}\theta)_{p}\bigl(\operatorname{grad} f(p)\bigr).
\]
Moreover, since $M$ is a dually flat space, we have 
$g^{ij}(p) = \frac{\partial \theta^i}{\partial \eta_j}(p)$, and thus
\[
\operatorname{grad} f(p)
= \sum_{i=1}^n \sum_{j=1}^n g^{ij}(p)\frac{\partial f}{\partial \theta^i}(p)
\Bigl(\frac{\partial}{\partial \theta^j}\Bigr)_p
= \sum_{j=1}^n \frac{\partial f}{\partial \eta_j}(p)
\Bigl(\frac{\partial}{\partial \theta^j}\Bigr)_p.
\]
Therefore it follows that 
$(\mathrm{d}\theta)_{p}\bigl(\operatorname{grad} f(p)\bigr)
= D_{\eta} f(p)$, and hence \eqref{egeo2} holds.
\end{proof}

\vspace{0.5cm}

When considering the minimization of a real-valued function 
$f$ on a convex set $\Theta \subset \mathbb{R}^m$, the mirror descent method 
employs the Bregman divergence
$B_{\psi}(p,q)
= \psi(p) - \psi(q) - (D\psi(q))^\top (p - q)$
associated with a convex function $\psi$ on $\Theta$.
The corresponding update rule is defined by
\begin{equation}
\label{mirrorstep}
\thetait{k+1}
= \argmin_{\theta \in \Theta}
\Bigl\{ (D f(\thetait{k}))^\top \theta
+ \frac{1}{t_k} B_{\psi}(\theta, \thetait{k}) \Bigr\}.
\end{equation}
Raskutti and Mukherjee~\cite{raskutti2015} showed that this update rule is 
equivalent to the natural gradient method
\begin{equation*}
\label{naturalstep}
\etait{k+1}
= \etait{k}
- t_k\, g(\etait{k})^{-1} D \bar{f}(\etait{k}),
\end{equation*}
where the reparametrized function 
$\bar{f}$ is defined by $\bar{f}(\eta) \coloneqq f(\theta(\eta))$ through the coordinate 
transformation $\theta_j(\eta) \coloneqq \frac{\partial \varphi}{\partial \eta_j}(\eta)$ 
for $1 \le j \le n$.

Within the framework of Lemma~\ref{双対平坦空間における測地線}, the result of 
Raskutti and Mukherjee~\cite{raskutti2015} can be reformulated to interpret 
mirror descent as a Riemannian steepest descent method along the geodesics of 
the dual connection.
We regard $\Theta$ as a dually flat space endowed with 
potential function $\psi$.
For an arbitrary point $\iter{p}{k}$ in a 
dually flat space $(\Theta, g, \nabla, \nabla^{*})$, there exists a 
neighborhood $U$ of $\iter{p}{k}$ together with affine coordinate systems 
$\theta$ for $\nabla$ and $\eta$ for $\nabla^{*}$. Let 
$\iter{\theta}{k} \coloneqq \theta(\iter{p}{k})$, and assume that $f$ is 
smooth on $\Theta$. We further assume that $\theta(U)$ is a convex set. Then, 
using the Bregman divergence
$B_{\psi}(\theta, \tilde{\theta})
= \psi(\theta) - \psi(\tilde{\theta})
- (D\psi(\tilde{\theta}))^\top(\theta - \tilde{\theta})$
associated with the potential $\psi$ in the 
$\theta$-coordinates,
we consider the mirror descent update~\eqref{mirrorstep}.

\vspace{0.5cm}

\begin{proposition}\label{dual-mirror}
Let $\bar{f} \coloneqq f \circ \theta$ be the function obtained by regarding $f$ as a function on $U$.
Consider the geodesic $\gamma^{*}$ with respect to $\nabla^{*}$ satisfying
$\gamma^{*}(0)=\iter{p}{k}$ and $\dot{\gamma}^{*}(0)=-\operatorname{grad}\bar{f}(\iter{p}{k})$.
For sufficiently small $\iter{t}{k}$, put $\iter{p}{k+1}=\gamma^{*}(\iter{t}{k})$.
Then $\iter{\theta}{k+1} = \theta(\iter{p}{k+1})$,
where $\iter{\theta}{k+1}$ is defined by \eqref{mirrorstep}.
\end{proposition}

\vspace{0.5cm}

Proposition \ref{dual-mirror} demonstrates that when the feasible set coincides with
$\theta(U)$ for some affine coordinate system $\theta$, and $\theta(U)$ is convex,
then for sufficiently small $t_k$, the mirror descent method is not only a natural
gradient method but also a descent method based on geodesics.

\begin{proof}[Proof of Proposition \ref{dual-mirror}]
The proof follows the argument in \cite{raskutti2015}.
Since $\eta$ is a $\nabla^{*}$-affine coordinate system, we have
\begin{align*}
\frac{\partial}{\partial \theta^i}
\biggl\{
\sum_j \theta^j \frac{\partial f}{\partial \theta^j}
(\thetait{k})
+ \frac{1}{t} B_\phi(\theta,\thetait{k})
\biggr\}
&=
\frac{\partial}{\partial \theta^i}
\biggl[
\sum_j \theta^j \frac{\partial f}{\partial \theta^j}
(\thetait{k})
+ \frac{1}{t}
\bigl\{ \psi(\theta) + \phi(\eta(\thetait{k}))
- \sum_j \theta^j \eta_j (\thetait{k}) \bigr\}
\biggr] \\
&=
\frac{\partial f}{\partial \theta^i}
(\thetait{k}) + \frac{1}{t}
\bigl\{ \eta_i - \eta_i(\thetait{k}) \bigr\},
\end{align*}
which completes the proof.
\end{proof}

\vspace{0.5cm}

In the natural gradient method, the direction in which the parameters are moved to
decrease the objective function is specified based on the Riemannian metric.
Although the descent trajectory of the natural gradient method may depend on the
choice of coordinate system, the mirror descent method has the structure that,
as we will see below, its descent trajectory coincides with a geodesic.

\vspace{0.5cm}

\begin{example}
\label{exponentiated}
The exponentiated gradient method can also be understood as a steepest descent method
with respect to the dual affine connection.
The exponentiated gradient update on the probability simplex
$\bar{\Delta}^n \coloneqq \{ y \in \mathbb{R}^n \mid y_i \ge 0,
\sum_{i=1}^n y_i = 1 \}$ for a smooth function
$f \colon \mathbb{R}^n \to \mathbb{R}$ is given by
\begin{equation}
\label{expo}
\iter{r}{k+1}_j
=
\frac{\iter{r}{k}_j
\exp\!\bigl(-t_k (D_{\mathbb{R}^n} f(\iter{r}{k}))_j\bigr)}
{\sum_{i=1}^n \iter{r}{k}_i
\exp\!\bigl(-t_k (D_{\mathbb{R}^n} f(\iter{r}{k}))_i\bigr)},
\qquad j = 1, \dots, n,
\end{equation}
as proposed in \cite{kivinen1997}.
Here $(D_{\mathbb{R}^n} f(\iter{r}{k}))_j$ denotes the usual partial derivative
with respect to the $j$th Euclidean coordinate.

It is known that the exponentiated gradient method can be formulated as
a mirror descent method (e.g., \cite[p.\,251]{beck2017}).
Let $\eta=(\eta_1,\dots,\eta_n)^\top \in
\Delta^{n} \coloneqq \{ \eta \in \mathbb{R}^n \mid
\eta_i > 0, \sum_{i=1}^n \eta_i = 1 \}$.
For $x \in \Omega \coloneqq \{1,\dots,n\}$, define $p(x \mid \eta) \coloneqq \eta_x$.
Let $S = \{ p(x \mid \eta) \mid \eta \in \Delta^n \}$.
Then $S$ is an exponential family as in \eqref{eq:likelihood3}.
Define $\bar{\eta} \colon S \to \Delta^n$ by
$\bar{\eta}(p(x \mid \eta)) = \eta$.

Let $\bar{f} \coloneqq f \circ \bar{\eta}$ be $f$ regarded as a function on $S$,
and for $\iter{r}{k} \in \Delta^n$, set
$\iter{p}{k} \coloneqq \bar{\eta}^{-1}(\iter{r}{k})$.
Consider the e-geodesic $\gamma^{(\mathrm{e})}$ on $S$
satisfying $\gamma^{(\mathrm{e})}(0)=\iter{p}{k}$ and
$\dot{\gamma}^{(\mathrm{e})}(0) = -\operatorname{grad}\bar{f}(\iter{p}{k})$.
Then for $\iter{p}{k+1} = \gamma^{(\mathrm{e})}(\iter{t}{k})$,
we obtain $\iter{r}{k+1} = \bar{\eta}(\iter{p}{k+1})$.

Since $\theta(S)=\mathbb{R}^{n-1}$, the manifold $S$ is geodesically complete
with respect to the e-connection, so no restriction is imposed on the step size $t_k$.
On the other hand, $S$ is not geodesically complete with respect to the Levi–Civita
connection or the m-connection.
Thus, the exponentiated gradient method can be seen as a steepest descent method
along geodesics that effectively exploit the geometric structure of $S$.
\end{example}

\vspace{0.5cm}

These two results not only provide a geometric interpretation of existing methods
but also indicate the usefulness of the dual-connection approach.

\section{Optimization of Objective Functions with Geometric Structure}
Optimization methods that utilize the structure of Riemannian manifolds are often effective
because there exists a certain relationship between the objective function and the underlying geometric structure.
For example, when the Fisher information matrix is used as the Riemannian metric,
the natural gradient method becomes essentially equivalent to the Fisher scoring algorithm introduced by Fisher \cite{fisher1925};
see, e.g., \cite[p.~245]{lange2013} for an exposition.
The Fisher scoring method has long been used in statistics, particularly for estimating maximum likelihood parameters.
In this case, its effectiveness is closely related to the intimate connection between the log-likelihood function and the Fisher metric.

In this section, we investigate the relationship between the minimization of objective functions based on divergences or log-likelihoods
and descent methods along geodesics induced by dual connections.

We consider the case where the objective function is the divergence 
on a dually flat space. As we will see below, this setting essentially 
includes the maximization of the log-likelihood function for an exponential family.

Fix a reference point $q \in U$. Let $D$ denote the canonical divergence 
on $U$ and define a function $f \colon U \to \mathbb{R}$ by
\begin{equation}
 f(r) \coloneqq B_\psi(r, q)
 = \psi(\theta(r)) + \varphi(\eta(q))
 - \sum_{j=1}^m \theta^j(r) \eta_j(q).
\end{equation}
The function $f(r)$ attains its minimum value $0$ when $r = q$.
In our optimization procedure, however, we do not make use of either
the fact that the minimizer is $q$ or that the minimum value is $0$.
\vspace{0.5cm}
\begin{theorem} \label{dualdiv}
For any $p, q \in U$,
\begin{equation}
\label{kousin}
\eta(q) = \eta(p) - (\mathrm{d}\eta)_p \, \operatorname{grad} f(p),
\end{equation}
where $f(r) \coloneqq B_\psi(r, q)$.
\end{theorem}

\vspace{0.5cm}

\begin{proof}
Since $g^{ij}(p) = \frac{\partial \theta^i}{\partial \eta_j}(p)$ and
$\frac{\partial f}{\partial \theta_j}(p)
= \frac{\partial \psi}{\partial \theta_j}(p) - \eta_j(q)$, we obtain
\begin{align}
\operatorname{grad} f(p)
&= \sum_{i=1}^m \sum_{j=1}^m g^{ij}(p)
   \frac{\partial f}{\partial \eta_i}(p)
   \left(\frac{\partial}{\partial \eta_j}\right)_p \notag \\
&= \sum_{j=1}^m \frac{\partial f}{\partial \theta^j}(p)
   \left(\frac{\partial}{\partial \eta_j}\right)_p \notag \\
&= \sum_{j=1}^m
   \left( \frac{\partial \psi}{\partial \theta^j}(p)
   - \eta_j(q) \right)
   \left(\frac{\partial}{\partial \eta_j}\right)_p \notag \\
&= \sum_{j=1}^m \left( \eta_j(p) - \eta_j(q) \right)
   \left(\frac{\partial}{\partial \eta_j}\right)_p.
\label{eq_grad}
\end{align}
Thus, $(\mathrm{d}\eta)_p \operatorname{grad} f(p)
= \eta(p) - \eta(q)$,
and this proves \eqref{kousin}.
\end{proof}

\vspace{0.5cm}

From Theorem~\ref{dualdiv}, we obtain the following corollary.
Let $(S, g, \nabla^{(\ee)}, \nabla^{(\mm)})$ be 
the statistical manifold induced by an exponential family.
Fix an arbitrary point $q \in S$. Then there exist affine coordinate
systems $\theta$ for $\nabla^{(\ee)}$ and $\eta$ for $\nabla^{(\mm)}$,
which are mutually dual with respect to $g$.
Using the Kullback--Leibler divergence, we define
\begin{equation}
 f(\eta) \coloneqq \mathrm{KL}(q, r(\eta)), ~~~
 h(\theta) \coloneqq \mathrm{KL}(r(\theta), q).
\end{equation}

\begin{corollary}
\label{kldiv}
For any $p, q \in S$, the following relations hold:
\begin{align*}
 \eta(q) &= \eta(p) - (\mathrm{d}\eta)_p \, \operatorname{grad} f(p), \\
 \theta(q) &= \theta(p) - (\mathrm{d}\theta)_p \, \operatorname{grad} h(p),
\end{align*}
where $f(\eta) \coloneqq \mathrm{KL}(q, r(\eta))$ and
$h(\theta) \coloneqq \mathrm{KL}(r(\theta), q)$.
\end{corollary}

\vspace{0.5cm}

\begin{proof}
The result follows from Theorem~\ref{dualdiv}, 
the existence of affine coordinate systems associated with 
$\nabla^{(\ee)}$ and $\nabla^{(\mm)}$ for an exponential family,
and the fact that, for an exponential family, 
the canonical divergence coincides with the 
Kullback--Leibler divergence.
\end{proof}

Combining this result with the update rules established in Section~2, 
we obtain the following interpretation.
If the $\mm$-geodesic $\gamma^{(\mm)}$ satisfying 
$\gamma^{(\mm)}(0) = p$ and 
$\dot{\gamma}^{(\mm)}(0) = - \operatorname{grad} f(p)$ 
is defined on $[0,1]$, then $\gamma^{(\mm)}(1)$ 
is the minimizer of $f$.
Similarly, if the $\ee$-geodesic $\gamma^{(\ee)}$ satisfying 
$\gamma^{(\ee)}(0) = p$ and 
$\dot{\gamma}^{(\ee)}(0) = - \operatorname{grad} h(p)$ 
is defined on $[0,1]$, then $\gamma^{(\ee)}(1)$ 
is the minimizer of $h$.

The following result holds for maximum likelihood estimation 
in an exponential family.

\begin{theorem}\label{MLE}
Let $S$ be the statistical manifold
associated with an exponential family 
with densities $p(x \mid \theta)$ for $\theta \in \theta(S)$.
Suppose that $N$ independent samples 
$x^N = (x(1), x(2), \dots, x(N))$ are drawn from a distribution 
in $S$, and assume that a maximum likelihood estimator 
$\hat{\theta} \in \theta(S)$ exists such that
$\sup_{\theta} p(x^N \mid \theta)
= p(x^N \mid \hat{\theta})$.
Define the negative log-likelihood function 
$f \colon S \to \mathbb{R}$ by
\begin{equation}
\label{nll}
f(p) \coloneqq - \sum_{l=1}^N \log p(x(l) \mid \theta(p)),
\end{equation}
where $p \in S$.
Let $\hat{p} \coloneqq p(\hat{\theta})$.
Then, for any $p \in S$,
\begin{equation}
\eta(\hat{p})
=
\eta(p)
- \frac{1}{N} (\mathrm{d}\eta)_p \,
\operatorname{grad} f(p).
\end{equation}
\end{theorem}
\vspace{0.5cm}
\begin{proof}
Let $\hat{p} \in S$ denote the point corresponding to the 
maximum likelihood estimator obtained from $x^N$.
From \eqref{nll} and the identity 
$\eta_k(\hat{p}) = \sum_{l=1}^N T_k(x(l))/N$
for $k = 1, \dots, m$, we have
\begin{align}
f(r)
&=
\sum_{l=1}^N \Bigl\{ \psi(\theta(r)) 
- \sum_{k=1}^m \theta^k(r) T_k(x(l)) \Bigr\} \notag \\
&=
N \Bigl[
\psi(\theta(r)) 
- \sum_{k=1}^m \theta^k(r)\, \eta_k(\hat{p})
+ \varphi(\eta(\hat{p}))
- \varphi(\eta(\hat{p}))
\Bigr] \notag \\
&=
N \Bigl\{
\mathrm{KL}(\hat{p}, r)
- \varphi(\eta(\hat{p}))
\Bigr\}.
\label{lldiv}
\end{align}
The claim follows from Corollary~\ref{kldiv}.
\end{proof}

From Theorem~\ref{MLE} together with 
Lemma~\ref{双対平坦空間における測地線},
it follows that, for any $p \in S$,
the $\mm$-geodesic $\gamma^{(\mm)}$ satisfying
$\gamma^{(\mm)}(0)=p$ and
$\dot{\gamma}^{(\mm)}(0)
= -\operatorname{grad} f(p)$
yields $\gamma^{(\mm)}(1/N)$ as a stationary point of $f$.
In other words, the steepest descent method along
$\mm$-geodesics with step size $1/N$
reaches the maximum likelihood estimator in a single step.

\section{Numerical Experiments}

We apply the optimization method developed in the previous section 
to several examples.

\begin{example}
We present numerical experiments consistent with 
Corollary~\ref{kldiv}. We consider the minimization of the 
Kullback--Leibler divergence between categorical distributions 
using update steps based on $\ee$- and $\mm$-geodesics.

We set the dimension $n = 3$ and initialize the step size as 
$t = 1.0$.
The step size $t$ is successively halved 
until the stopping condition is satisfied.
The initial point is chosen as 
$\eta_0 = (1/3,\, 1/3) \in S_{2}$.

We generate the target distribution $q \in S$ by sampling
$q'_1, q'_2, q'_3$ independently from the uniform distribution 
on $[0,1]$ and defining
\[
q_1 = \frac{q'_1}{q'_1 + q'_2 + q'_3}, \qquad
q_2 = \frac{q'_2}{q'_1 + q'_2 + q'_3}.
\]
We then set $q = \eta^{-1}(q_1, q_2)$.

We define the objective functions
\begin{align}
    f(r) &= \mathrm{KL}(q, r), \\
    h(r) &= \mathrm{KL}(r, q).
\end{align}

We apply update steps based on $\ee$-, $\mm$-, and Levi-Civita
geodesics to minimize $f$ and $h$ over $S$.
Since the minimum of both $f$ and $h$ is attained at $q$, 
we terminate the iteration when 
$\|\eta(p^{(k)}) - \eta(q)\|_{2} < \varepsilon$
with $\varepsilon = 10^{-5}$.

The convergence behavior may depend on the values of 
$q'_1, q'_2, q'_3$. Therefore, we repeat the experiment 
100 times with independently generated $q$ and compute the mean 
and standard deviation of the number of iterations required for convergence.
The results rounded to three decimal places
are shown in Tables~\ref{tab:forwardkl} and \ref{tab:backkl}.

\begin{table}[H]
\centering
\doublespacing
\caption{Mean and standard deviation of iterations required 
for minimizing $f$ for a categorical distribution ($n=3$)}
\label{tab:forwardkl}
\begin{tabular}{|c|c|c|}
\hline
Update rule & Mean iterations & Standard deviation \\
\hline
$\mm$-connection &  1.00 & 0.00 \\
$\ee$-connection & 3.66 & 1.06 \\
\hline
\end{tabular}
\end{table}

\begin{table}[H]
\centering
\doublespacing
\caption{Mean and standard deviation of iterations required 
for minimizing $h$ for a categorical distribution ($n=3$)}
\label{tab:backkl}
\begin{tabular}{|c|c|c|}
\hline
Update rule & Mean iterations & Standard deviation \\
\hline
$\mm$-connection &  3.84 & 1.09 \\
$\ee$-connection & 1.00 & 0.00 \\
\hline
\end{tabular}
\end{table}
For the minimization of $f$, the $\ee$-geodesic update 
converged in exactly one iteration for all 100 trials, 
while for the minimization of $h$, the $\mm$-geodesic update 
converged in exactly one iteration for all 100 trials.
These results are consistent with Corollary~\ref{kldiv}.

Furthermore, for a categorical distribution, 
maximum likelihood estimation is equivalent to the minimization 
of $f$ by \eqref{lldiv}. Therefore, as shown in 
Table~\ref{tab:forwardkl}, using the update along 
$\mm$-geodesics enables convergence to the maximum likelihood 
estimator in a single iteration.
\end{example}

\vspace{0.5cm}

\begin{example}
We conducted similar numerical experiments on maximum likelihood estimation for mixture families.
Using the given probability distributions $p_1(x), \dots, p_n(x)$ we define
\[
p(x \mid \eta) = \sum_{k=1}^{n-1} \eta_k p_k(x)
    + \left( 1 - \sum_{k=1}^{n-1} \eta_k \right) p_n(x)
\]
The mixture family $\mathcal{S}$ is then defined as
\[
S = \left\{ p(x \mid \eta) \ \middle| \ \eta \in S_{n-1} \right\},
\]
where $S_{n-1}$ is defined by \eqref{multi}.
It is known that the mixture family is a dually flat space.

We set $n = 4$ and $\Omega = \{0, \dots, 7\}$.
The distributions $p_1(x), p_2(x), p_3(x), p_4(x)$ are uniform on
$A_1 = \{0, 1, 2\}$, $A_2 = \{2, 3, 4\}$,
$A_3 = \{4, 5, 6\}$, and $A_4 = \{6, 7, 0\}$, respectively.
The parameter $\eta$ is the m-affine coordinate system.
The e-affine coordinate system dual to the m-affine coordinate system is given by
\begin{align}
\label{thetaeta}
\theta^i = \sum_x
\bigl( p_i(x) - p_n(x) \bigr) \log p(x \mid \eta).
\end{align}
It can be shown that the parameter space of the affine coordinate system
$\theta$ with respect to the e-connection is
$\theta(S) = \mathbb{R}^3$.

In this setting, the transformation from the m-affine coordinate $\eta$ 
to the e-affine coordinate $\theta$
was computed by \eqref{thetaeta},
and the inverse transformation from $\theta$ to $\eta$ was numerically computed
using the Newton method.

We compared three methods
the m-geodesic method, the e-geodesic method, and the exponentiated gradient method.
The stopping criterion was that the Euclidean norm of the gradient of the
negative log-likelihood function with respect to $(\eta_1, \eta_2, \eta_3)$
becomes less than $10^{-5}$ and we compared the number of iterations.
The initial value was set to
$(\eta_1, \eta_2, \eta_3) = (\frac{1}{4}, \frac{1}{4}, \frac{1}{4})$.

The sample size was $N = 1000$ and we considered three cases
(1) 250 samples were randomly drawn from each of $A_1, A_2, A_3, A_4$
(2) 400 samples were randomly drawn from each of $A_1$ and $A_2$
and 100 samples were randomly drawn from each of $A_3$ and $A_4$
(3) 700 samples were randomly drawn from $A_1$
and 100 samples were randomly drawn from each of $A_2, A_3$, and $A_4$.

For the step size $lr$, in the m-geodesic method the update rule halved the value
until the updated parameter entered the parameter space $S_3$ of $\eta$.
Taking a large $lr$ causes the Newton method used in the e-geodesic computation
to fail to converge therefore we examined three cases
$lr = 0.5/N, 1.0/N, 1.5/N$.

The results are summarized in Tables \ref{mix_250}, \ref{mix_400}, and \ref{mix_700}.

\begin{table}[H]
    \centering
    \doublespacing
    \caption{Number of iterations for the data configuration (250, 250, 250, 250)}
    \label{mix_250}
    \begin{tabular}{|c|c|c|c|}
        \hline
        Method $\;\setminus\;$ lr & $\frac{0.5}{N}$ & $\frac{1.0}{N}$ & $\frac{1.5}{N}$ \\
        \hline
        Exponentiated Gradient & 89 & 40 & 24 \\
        m-Geodesic & 25 & 5 & 22 \\
        e-Geodesic & 27 & 5 & 23 \\
        \hline
    \end{tabular}
\end{table}

\begin{table}[H]
    \centering
    \doublespacing
    \caption{Number of iterations for the data configuration (400, 400, 100, 100)}
    \label{mix_400}
    \begin{tabular}{|c|c|c|c|}
        \hline
        Method  $\;\setminus\;$ lr & $\frac{0.5}{N}$ & $\frac{1.0}{N}$ & $\frac{1.5}{N}$ \\
        \hline
        Exponentiated Gradient & 82 & 37 & 21 \\
        m-Geodesic & 27 & 9 & 49 \\
        e-Geodesic & 28 & 9 & 47 \\
        \hline
    \end{tabular}
\end{table}

\begin{table}[H]
    \centering
    \doublespacing
    \caption{Number of iterations for the data configuration (700, 100, 100, 100)}
    \label{mix_700}
    \begin{tabular}{|c|c|c|c|}
        \hline
        Method $\;\setminus\;$ lr & $\frac{0.5}{N}$ & $\frac{1.0}{N}$ & $\frac{1.5}{N}$ \\
        \hline
        Exponentiated Gradient & 92 & 41 & 23 \\
        m-Geodesic & 29 & 8 & 41 \\
        e-Geodesic & 29 & 8 & 39 \\
        \hline
    \end{tabular}
\end{table}

Unlike the case of exponential families, we can observe that
the m-geodesic method with $lr = \frac{1}{N}$ does not converge in a single iteration,
and the number of iterations required for convergence is similar to that of the e-geodesic method.
In all cases regardless of the data configuration,
$lr = \frac{1}{N}$ achieves the fastest convergence for both the m- and e-geodesic methods.
Furthermore, the results of the Exponentiated Gradient method with a larger $lr$
are summarized in Table~\ref{expo_large}.

\begin{table}[H]
    \centering
    \doublespacing
    \caption{Number of iterations of the Exponentiated Gradient method for each data configuration}
    \label{expo_large}
    \begin{tabular}{|c|c|c|c|c|c|c|}
        \hline
        Data configuration $\;\setminus\;$ lr & $\frac{1.6}{N}$ & $\frac{1.7}{N}$ & $\frac{1.8}{N}$ & $\frac{1.9}{N}$ & $\frac{2.0}{N}$ & $\frac{2.1}{N}$ \\
        \hline
        $(250, 250, 250, 250)$ & 22 & 20 & 18 & 17 & 17 & 20 \\
        $(400, 400, 100, 100)$ & 19 & 17 & 16 & 16 & 19 & 23 \\
        $(700, 100, 100, 100)$ & 21 & 19 & 24 & 31 & 40 & 57 \\
        \hline
    \end{tabular}
\end{table}

Based on these results and Example~\ref{exponentiated},
given that the Exponentiated Gradient method corresponds to
the $e$-geodesic on the categorical distribution over the probability simplex,
we can observe that the $e$- and $m$-geodesic methods,
which correctly exploit the structure of the probability distributions
of mixture families,
require fewer iterations than the Exponentiated Gradient method.
\end{example}

\vspace{0.5cm}

\begin{example}
We consider the Bradley--Terry model \cite{BA1952},
which is a practically important application of maximum likelihood estimation
in exponential families.
The problem setting of the Bradley--Terry model is briefly explained as follows.
Let there be $N$ players and suppose that each player $i$ has a strength
parameter $\pi_i > 0$ satisfying $\sum_{i=1}^N \pi_i = 1$.
The probability that player $i$ defeats player $j$ is defined as
$p_{ij} = \pi_i / (\pi_i + \pi_j)$.

Let $n_{ij}$ $(i,j = 1,\ldots, N)$ denote the number of matches between players $i$ and $j$,
and let $x_{ij} \in \{0,1,\ldots,n_{ij}\}$ denote the number of times
that $i$ wins against $j$.
In particular, $x_{ij} + x_{ji} = n_{ij}$ holds for all $i \neq j$.
We denote the collection of all observed match results by
$x = (x_{ij})_{1 \le i < j \le N}$.
Then the probability of observing $x$ is given by
\[
\prod_{i=1}^{N-1} \prod_{j=i+1}^{N}
\binom{n_{ij}}{x_{ij}}
p_{ij}^{x_{ij}}
p_{ji}^{n_{ij} - x_{ij}}.
\]

We define the sufficient statistics and natural parameters by
\begin{align}
T_i(x) &\coloneqq \sum_{j \neq i} x_{ij}, \\
\theta^i &\coloneqq \log \frac{\pi_i}{\pi_N}
\quad (i = 1, \ldots, N), \quad
\theta = (\theta^1, \ldots, \theta^{N-1}), \\
\psi(\theta) &\coloneqq \sum_{1 \le i < j \le N}
n_{ij} \log \left( \exp(\theta^i) + \exp(\theta^j) \right).
\end{align}
Then the probability density with respect to the base measure
$\prod_{i=1}^{N-1} \prod_{j=i+1}^{N}
\binom{n_{ij}}{x_{ij}}$
is given by
\[
p(x \mid \theta)
=
\prod_{i=1}^{N-1} \prod_{j=i+1}^{N}
p_{ij}^{x_{ij}}
p_{ji}^{n_{ij} - x_{ij}}
=
\exp \left(
\sum_{k=1}^{N-1} \theta^k T_k(x)
\;-\;
\psi(\theta)
\right).
\]
Therefore, the Bradley--Terry model belongs to an exponential family.

Based on the observation $x$, the maximum likelihood estimator can be obtained
by finding $\theta \in \mathbb{R}^{N-1}$ that maximizes $p(x \mid \theta)$.
Since this model belongs to an exponential family,
the maximum likelihood estimator of the expectation parameters is
$\hat{\eta}_i = T_i(x)$.
However, it is difficult to explicitly express the transformation from $\eta$ to $\pi$.

On the other hand, for the e-affine coordinate
$\theta$, we can compute $\hat{\pi}$ from the maximum likelihood estimate
$\hat{\theta}$ by
\[
\pi_i = \frac{\exp(\theta^i)}
{1 + \sum_{k=1}^{N-1} \exp(\theta^k)}
\quad (i = 1, \dots, N-1).
\]
Moreover, since $\theta(S) = \mathbb{R}^{N-1}$,
there is no need to consider constraints on the parameter space.
Therefore, we perform optimization using e-geodesics.

For comparison, in the maximum likelihood estimation for the Bradley--Terry model,
we compare the number of iterations of
the MM algorithm and the exponentiated gradient method.
Starting from the initial value
$\left(\frac{1}{N}, \dots, \frac{1}{N}\right)$,
we iterate until the Euclidean norm of the gradient of the negative
log-likelihood function with respect to
$\pi_1, \dots, \pi_{N-1}$ becomes less than $10^{-5}$.
A commonly used MM algorithm
(e.g.\ \cite{lange2013}, p.\ 193) updates the parameters as follows.
\begin{enumerate}
    \item Set an initial estimate $\hat{\pi}_i^{(0)}$.
    \item Compute $\tilde{\pi}_i^{(l+1)}$ by
    \[
    \tilde{\pi}_i^{(l+1)}
    =
    \frac{\sum_{j \in \{1,\dots, N\} \setminus \{i\}} x_{ij}}
    {\sum_{j \in \{1,\dots, N\} \setminus \{i\}}
        \frac{n_{ij}}
        {\hat{\pi}_i^{(l)} + \hat{\pi}_j^{(l)}}}.
    \]
    \item Normalize
    \[
    \hat{\pi}_i^{(l+1)}
    = \frac{\tilde{\pi}_i^{(l+1)}}
    {\sum_{k=1}^N \tilde{\pi}_k^{(l+1)}}.
    \]
\end{enumerate}

We describe the setting of the numerical experiment.
We first consider the case $N = 3$ with
$x_{12} = 7, \
x_{13} = 8, \
x_{21} = 3, \
x_{23} = 5, \
x_{31} = 2, \
x_{32} = 5$.
The variable $lr$ denotes the step size for the exponentiated gradient method
and the e-geodesic method.
The MM algorithm does not use a step size.
When $lr = 0.01$, the results are shown in Table~\ref{BT_small_0.01}.
First, when $lr = 0.01$, the results are shown in Table~\ref{BT_small_0.01}.
\begin{table}[H]
    \centering
    \doublespacing
    \caption{Comparison of the number of iterations for each method ($lr = 0.01$)}
    \label{BT_small_0.01}
    \begin{tabular}{|c|c|}
        \hline
        Method & Number of iterations \\
        \hline
        MM algorithm & 20 \\
        Exponentiated Gradient & 84 \\
        e-Geodesic & 1468 \\
        \hline
    \end{tabular}
\end{table}
On the other hand, when $lr = 1.0$, the Exponentiated Gradient method suffers
from numerical overflow and fails to compute the estimate,
whereas the e-geodesic method significantly reduces the number of iterations,
as shown in Table~\ref{BT_small_1.0}.
\begin{table}[H]
    \centering
    \doublespacing
    \caption{Comparison of the number of iterations for each method ($lr = 1.0$)}
    \label{BT_small_1.0}
    \begin{tabular}{|c|c|}
        \hline
        Method & Number of iterations \\
        \hline
        MM algorithm & 20 \\
        Exponentiated Gradient & overflow \\
        e-Geodesic & 4 \\
        \hline
    \end{tabular}
\end{table}
From Table~\ref{BT_small_1.0}, we observe that the e-geodesic method converges faster
when $lr = 1.0$.
Next, we consider larger values of $N$ with $lr = 1.0$.
Since the Exponentiated Gradient method causes numerical overflow in this setting,
we compare only the MM algorithm and the e-geodesic method.

We set $N = 100$.
For each $1 \le i < j \le N$, the values of $n_{ij}$ are sampled from the uniform
distribution on $\{1, \dots, 1000\}$,
and the values of $x_{ij}$ are sampled from the uniform distribution
on $\{0, \dots, n_{ij}\}$.
To account for randomness, we generate 100 instances following the above setting,
and compute the average and standard deviation of the number of iterations required for convergence,
where convergence is defined as the Euclidean norm of the gradient
of the negative log-likelihood function with respect to
$\pi_1, \dots, \pi_{N-1}$ falling below $10^{-5}$.
The standard deviation values are rounded to the third decimal place.
The results are shown in Table~\ref{BT_large}.
\begin{table}[H]
    \centering
    \doublespacing
    \caption{Average and standard deviation of the number of iterations required for MLE in the Bradley--Terry model with $N = 100$ ($lr = 1.0$)}
    \label{BT_large}
    \begin{tabular}{|c|c|c|}
            \hline
            Method & Average iterations & Standard deviation \\
            \hline
            MM algorithm & 46.31 & 1.76 \\
            e-Geodesic & 3.11 & 0.31 \\
            \hline
    \end{tabular}
\end{table}
These results show that as $N$ increases,
the MM algorithm requires more iterations.
In contrast, the e-geodesic method still converges in only a small number of steps.
\end{example}

\vspace{0.5cm}

\begin{example}
\label{MLR}
As an example that utilizes the result of Corollary~\ref{kldiv}
for $h(\theta) \coloneqq \mathrm{KL}(\theta, q)$,
we consider variational inference for multinomial logistic regression.

Following Blundell et al.\ \cite{blundell2015} and Suyama \cite{suyama2017},
we describe the setting of multinomial logistic regression
and the application of gradient descent.

Let $\mathcal{Z} \coloneqq \{ z \in \{0,1\}^D \mid \sum_{i=1}^D z_i = 1 \}$.
For each explanatory variable vector $x \in \mathbb{R}^M$,
we assume that a response variable $y \in \mathcal{Z}$ is generated stochastically.

Suppose that we obtain $N$ observations $(x_1, y_1), \dots, (x_N, y_N)$, and define
\[
X =
\begin{pmatrix}
x_1^\top \\ \vdots \\ x_N^\top
\end{pmatrix},
\quad
Y =
\begin{pmatrix}
y_1^\top \\ \vdots \\ y_N^\top
\end{pmatrix}.
\]

In multinomial logistic regression, for a parameter matrix
$W \in \mathbb{R}^{M \times D}$,
the conditional distribution of $y_i$ given $x_i$ is modeled as
\begin{align}
p(Y \mid X, W)
= \prod_{i=1}^N \mathrm{Cat}\!\Bigl(
y_i \,\Big|\, \mathrm{SM}(W^\top x_i)
\Bigr),
\end{align}
where, for $a = (a_1, \dots, a_D)^\top \in \mathbb{R}^D$,
\[
\mathrm{SM}(a)
= \bigl(\mathrm{SM}_1(a), \dots, \mathrm{SM}_D(a)\bigr)^\top,
\]
with
\[
\mathrm{SM}_k(a)
= \frac{\exp(a_k)}{\sum_{j=1}^D \exp(a_j)},
\quad k = 1, \dots, D,
\]
and
\[
\mathrm{Cat}(y \mid s)
= \prod_{k=1}^D s_k^{\,y_k}.
\]

As a prior distribution for $W$, we assume independent Gaussian distributions
\begin{align}
p(W)
= \prod_{i=1}^{M} \prod_{j=1}^{D}
\mathcal{N}\!\left(w_{ij} \mid 0,\, \lambda^{-1}\right)
\end{align}
with precision parameter $\lambda > 0$.

Since exact Bayesian inference for $W$ is intractable,
we consider a variational approximation in which each element of $W$
is modeled as an independent Gaussian
\begin{align}
q(W \mid \mu, \sigma)
= \prod_{i=1}^{M} \prod_{j=1}^{D}
\mathcal{N}\!\left(w_{ij} \mid \mu_{ij},\, \sigma_{ij}^2\right),
\end{align}
where
$\mu = (\mu_{11}, \dots, \mu_{MD})^\top$
and
$\sigma = (\sigma_{11}, \dots, \sigma_{MD})^\top$.
The variational parameters $\mu$ and $\sigma$ are optimized by minimizing
\begin{align}
\mathrm{KL}\!\left(
q(W \mid \mu, \sigma),\;
p(W \mid X, Y)
\right).
\label{logkl}
\end{align}

Expanding the Kullback--Leibler divergence gives
\begin{align}
\mathrm{KL}\!\left(q(W \mid \mu, \sigma),\, p(W \mid X, Y)\right)
&=
\int \left\{
\log q(W \mid \mu, \sigma)
- \log p(Y \mid X, W)
- \log p(W)
\right\} q(W \mid \mu, \sigma)\, \mathrm{d}W
\notag \\
&\quad + \log p(Y \mid X),
\label{logkl2}
\end{align}
where $\log p(Y \mid X)$ is a constant independent of $W$.

Since the analytic computation of the gradient of \eqref{logkl2} is difficult,
we approximate the gradient using Monte Carlo integration.
Let $\Sigma = \mathrm{diag}(\sigma_1, \dots, \sigma_{MD})$
and let $\epsilon_k \sim \mathcal{N}(0, I_{MD})$.
We set
\[
\mathrm{vec}(W_k) = \mu + \Sigma \epsilon_k,
\]
and approximate the first term on the right-hand side of \eqref{logkl2} by
\begin{align}
\label{logkl4}
h(\mu, \sigma)
\coloneqq
\frac{1}{K} \sum_{k=1}^{K}
\left\{
\log q(W_k \mid \mu, \sigma)
-
\log p(Y \mid X, W_k)
-
\log p(W_k)
\right\},
\end{align}
and then compute its gradient.
The dependence on $X$ is omitted since $X$ is fixed.

It has long been common in statistics and econometrics to compute approximate
maximum-likelihood estimators and moment estimators by holding fixed
the realizations of random numbers across parameter values
(see e.g.\ \cite{GM1997}).
Recently, this method has been referred to as the
\emph{reparameterization trick} in machine learning
\cite{kingma2014, rezende2014}.

Noting that $W_k$ depends on $\mu$ and $\sigma$ through
$\mathrm{vec}(W_k) = \mu + \Sigma \epsilon_k$,
we compute the gradient of \eqref{logkl4}
and apply gradient descent to minimize $h(\mu, \sigma)$.

To avoid the positivity constraint $\sigma_j > 0$,
we introduce parameters $\rho_j$ such that
$\sigma_j = \log(1 + \exp(\rho_j))$,
and apply gradient descent to $\mu$ and $\rho$.

After obtaining the parameter estimates $\hat{\mu}$ and $\hat{\sigma}$,
we predict the response $y' \in \mathcal{Z}$ for a new
explanatory variable $x' \in \mathbb{R}^M$
by Monte Carlo approximation.
We draw samples
$W_l \sim q(W \mid \hat{\mu}, \hat{\sigma})$
for $l = 1, \dots, L$,
and determine the prediction by
\begin{align}
\label{pred}
i \in \underset{j \in \{1,\dots,D\}}{\operatorname{argmax}}
\left(
\frac{1}{L}
\sum_{l=1}^{L}
\mathrm{SM}_j(W_l^\top x')
\right),
\end{align}
and set
\begin{align}
\label{intkl3}
y' = e_i,
\end{align}
where $e_i$ is the $i$-th standard basis vector of $\mathbb{R}^D$
with its $i$-th element equal to $1$ and the others equal to $0$.

Next, we describe the optimization using the e- and m-geodesics.
Let $S$ be the family of multivariate normal distributions
with diagonal covariance matrices.
As described in Section~3, $S$ is an exponential family.
Using $S$, \eqref{logkl} can be reformulated as
\begin{align}
\label{logkl5}
\min_{r \in S}
\mathrm{KL}(r,\, p(W \mid X, Y)).
\end{align}

This problem is approximated by using \eqref{logkl4} as
\begin{align}
\label{logkl5a}
\min_{r \in S} \bar{h}(r),
\end{align}
where $\bar{h}(r) = h(\mu(r), \sigma(r))$,
and $\mu(q(W \mid \bar{\mu}, \bar{\sigma})) = \bar{\mu}$,
$\sigma(q(W \mid \bar{\mu}, \bar{\sigma})) = \bar{\sigma}$.

Let $q^* \in S$ be the solution of \eqref{logkl5}.
Since $q^*$ is close to $p(W \mid X, Y)$,
we define
$f(r) \coloneqq \mathrm{KL}(r,\, q^*)$,
and we have
$\operatorname{grad} f(r)
\simeq
\operatorname{grad}\mathrm{KL}(r,\, p(W \mid X, Y))$.
Together with
$\operatorname{grad}\bar{h}(r)
\simeq
\operatorname{grad}\mathrm{KL}(r,\, p(W \mid X, Y))$,
it follows that
$\operatorname{grad}\bar{h}(r)
\simeq
\operatorname{grad} f(r)$.

Now take an arbitrary $p \in S$.
From Corollary~\ref{kldiv}, if an e-geodesic
$\gamma^{(\mathrm{e})}$ defined on $[0,1]$
satisfies
$\gamma^{(\mathrm{e})}(0) = p$,
$\dot{\gamma}^{(\mathrm{e})}(0) = -\operatorname{grad} g(p)$,
then $\gamma^{(\mathrm{e})}(1) = q^*$ holds.
Therefore, for any $p \in S$, if an e-geodesic $\gamma^{(\mathrm{e})}$
defined on $[0,1]$ satisfies
$\gamma^{(\mathrm{e})}(0) = p$ and
$\dot{\gamma}^{(\mathrm{e})}(0) = -\operatorname{grad}\bar{h}(p)$,
then we expect that
$
\gamma^{(\mathrm{e})}(1)
\simeq
q^*$.
This suggests that an e-geodesic update with step size $1$
can achieve a good approximation in a single iteration.

The e-affine coordinate system $\theta$ on $S$ is given, for each $p \in S$, by
\begin{align*}
\theta(p)
=
\left(
\frac{\mu_1}{\sigma_1^2}, \dots, \frac{\mu_{MD}}{\sigma_{MD}^2},
-\frac{1}{2\sigma_1^2}, \dots, -\frac{1}{2\sigma_{MD}^2}
\right),
\stepcounter{equation}\tag{\theequation}
\end{align*}
where $\mu = (\mu_1,\dots,\mu_{MD})$ and $\sigma = (\sigma_1,\dots,\sigma_{MD})$.
The m-affine coordinate system is given by
\begin{align*}
\eta(p)
=
\left(
\mu_1,\dots,\mu_{MD},
\mu_1^2+\sigma_1^2,\dots,\mu_{MD}^2+\sigma_{MD}^2
\right).
\stepcounter{equation}\tag{\theequation}
\end{align*}

It can be confirmed that both $\theta(S)$ and $\eta(S)$ are convex.
In the case of the diagonal Gaussian family $S$, as in the categorical distribution case,
we halve the step size $t$ repeatedly
until $\eta(p) - t \nabla_{\theta} f(p) \in \eta(S)$ is satisfied for the m-geodesic update,
and until $\theta(p) - t \nabla_{\eta} f(p) \in \theta(S)$ is satisfied for the e-geodesic update.

Next, we describe the settings of the numerical experiments.
A dataset $(X, Y)$ is randomly generated,
and the details of the generation procedure are provided in the appendix.

Among the generated pairs $(X, Y)$, $30\%$ are used as test data,
and the remaining $70\%$ are used as training data.
The initial values of $\mu$ and $\rho$ are generated independently for each component
from the standard normal distribution.
We set $L = 10$ and perform only a single iteration of the optimization.
We then compute the prediction $Y^*$ according to the procedure
and evaluate the accuracy by
$(\text{number of correct predictions})/N$.
We refer to this value as the accuracy.

Since the performance may depend on the initialization of $(\mu, \rho)$
and on the generated dataset $(X, Y)$,
we repeat the procedure $100$ times
and compute the mean and standard deviation of the accuracy
for both the training data and the test data.

We examine the following settings.
\[
(N, M, D) \in \{(200, 5, 3),\ (1000, 10, 2)\},
\quad
lr \in \{0.01, 1, 100\},
\quad
\lambda \in \{0.01, 1, 100\}.
\]
There are $18$ combinations in total.
If the chosen step size $lr$ does not satisfy the sufficient condition
for extending the geodesic as described above,
we repeatedly halve the step size \,$lr \leftarrow lr/2$\, until the condition is satisfied.
The number of random samples used for Monte Carlo approximation is set to $K = 10000$.

The results are shown in Tables~\ref{fo_table} and \ref{st_table}.
The accuracy and the standard deviation are rounded to the third decimal place.
The gradient descent method in $(\mu,\rho)$ is denoted by ``gradient'',
the e-geodesic method by ``e-geodesic'',
and the m-geodesic method by ``m-geodesic''.
\begin{table}[H]
    \renewcommand{\arraystretch}{0.8}
    \centering
    \doublespacing
    \caption{Mean accuracy and standard deviation for $K=10000$ and $(N, M, D) = (200, 5, 3)$. 
    The initial step size $lr$ is halved until a valid geodesic is defined.}
    \label{fo_table}
    \begin{tabular}{|c|c|c|c|c|c|c|}
    \hline
    $\lambda$ & $lr$  & Method & \multicolumn{2}{c|}{Training data} & \multicolumn{2}{c|}{Test data} \\ \cline{4-7} 
     & & & Mean & Std.\ dev. & Mean & Std.\ dev. \\ \hline
    \multirow{3}{*}{0.01} & 0.01 & gradient & 0.78 & 0.10 & 0.75 & 0.12 \\ 
     & & e-geodesic & 0.63 & 0.17 & 0.60 & 0.17 \\ 
     & & m-geodesic & 0.53 & 0.17 & 0.51 & 0.18 \\ \hline
    \multirow{3}{*}{0.01} & 1.0 & gradient & 0.80 & 0.11 & 0.79 & 0.11 \\ 
     & & \underline{e-geodesic} & 0.82 & 0.13 & 0.81 & 0.13 \\ 
     & & m-geodesic & 0.52 & 0.18 & 0.51 & 0.19 \\ \hline
    \multirow{3}{*}{0.01} & 100.0 & gradient & 0.80 & 0.11 & 0.79 & 0.11 \\ 
     & & e-geodesic & 0.78 & 0.13 & 0.78 & 0.14 \\ 
     & & m-geodesic & 0.52 & 0.18 & 0.50 & 0.18 \\ \hline\hline
    \multirow{3}{*}{1.0} & 0.01 & gradient & 0.78 & 0.10 & 0.75 & 0.12 \\ 
     & & e-geodesic & 0.63 & 0.17 & 0.60 & 0.17 \\ 
     & & m-geodesic & 0.53 & 0.18 & 0.50 & 0.18 \\ \hline
    \multirow{3}{*}{1.0} & 1.0 & gradient & 0.80 & 0.11 & 0.79 & 0.11 \\ 
     & & \underline{e-geodesic} & 0.82 & 0.12 & 0.81 & 0.13 \\ 
     & & m-geodesic & 0.52 & 0.18 & 0.50 & 0.19 \\ \hline
    \multirow{3}{*}{1.0} & 100.0 & gradient & 0.79 & 0.12 & 0.79 & 0.11 \\ 
     & & e-geodesic & 0.78 & 0.13 & 0.78 & 0.13 \\ 
     & & m-geodesic & 0.52 & 0.19 & 0.50 & 0.19 \\ \hline\hline
    \multirow{3}{*}{100.0} & 0.01 & gradient & 0.78 & 0.12 & 0.78 & 0.11 \\ 
     & & e-geodesic & 0.56 & 0.18 & 0.54 & 0.17 \\ 
     & & m-geodesic & 0.43 & 0.18 & 0.42 & 0.18 \\ \hline
    \multirow{3}{*}{100.0} & 1.0 & gradient & 0.51 & 0.18 & 0.52 & 0.18 \\ 
     & & \underline{e-geodesic} & 0.84 & 0.10 & 0.83 & 0.10 \\ 
     & & m-geodesic & 0.44 & 0.18 & 0.42 & 0.18 \\ \hline
    \multirow{3}{*}{100.0} & 100.0 & gradient & 0.51 & 0.18 & 0.52 & 0.18 \\ 
     & & e-geodesic & 0.79 & 0.14 & 0.79 & 0.14 \\ 
     & & m-geodesic & 0.43 & 0.18 & 0.42 & 0.18 \\ \hline
    \end{tabular}
\end{table}
\begin{table}[H]
    \renewcommand{\arraystretch}{0.8}
    \centering
    \doublespacing
    \caption{Mean accuracy and standard deviation for $K=10000$ and $(N, M, D) = (500, 3, 7)$.
    The initial step size $lr$ is halved until a valid geodesic is defined.}
    \label{st_table}
    \begin{tabular}{|c|c|c|c|c|c|c|}
    \hline
    $\lambda$ & $lr$ & Method & \multicolumn{2}{c|}{Training data} & \multicolumn{2}{c|}{Test data} \\ \cline{4-7}
     & & & Mean & Std.\ dev. & Mean & Std.\ dev. \\ \hline
    \multirow{3}{*}{0.01} & 0.01 & gradient & 0.60 & 0.09 & 0.58 & 0.10 \\
     & & e-geodesic & 0.33 & 0.13 & 0.32 & 0.13 \\
     & & m-geodesic & 0.21 & 0.11 & 0.21 & 0.11 \\ \hline
    \multirow{3}{*}{0.01} & 1.0 & gradient & 0.57 & 0.08 & 0.56 & 0.08 \\
     & & \underline{e-geodesic} & 0.75 & 0.07 & 0.74 & 0.08 \\
     & & m-geodesic & 0.21 & 0.11 & 0.21 & 0.11 \\ \hline
    \multirow{3}{*}{0.01} & 100.0 & gradient & 0.57 & 0.08 & 0.55 & 0.08 \\
     & & e-geodesic & 0.65 & 0.13 & 0.64 & 0.15 \\
     & & m-geodesic & 0.21 & 0.11 & 0.22 & 0.11 \\ \hline\hline
    \multirow{3}{*}{1.0} & 0.01 & gradient & 0.60 & 0.09 & 0.58 & 0.10 \\
     & & e-geodesic & 0.33 & 0.13 & 0.32 & 0.13 \\
     & & m-geodesic & 0.21 & 0.11 & 0.21 & 0.11 \\ \hline
    \multirow{3}{*}{1.0} & 1.0 & gradient & 0.57 & 0.08 & 0.55 & 0.08 \\
     & & \underline{e-geodesic} & 0.76 & 0.07 & 0.74 & 0.08 \\
     & & m-geodesic & 0.21 & 0.11 & 0.21 & 0.11 \\ \hline
    \multirow{3}{*}{1.0} & 100.0 & gradient & 0.56 & 0.08 & 0.55 & 0.08 \\
     & & e-geodesic & 0.66 & 0.13 & 0.65 & 0.14 \\
     & & m-geodesic & 0.21 & 0.11 & 0.22 & 0.11 \\ \hline\hline
    \multirow{3}{*}{100.0} & 0.01 & gradient & 0.56 & 0.08 & 0.55 & 0.09 \\
     & & e-geodesic & 0.28 & 0.13 & 0.28 & 0.13 \\
     & & m-geodesic & 0.17 & 0.10 & 0.17 & 0.10 \\ \hline
    \multirow{3}{*}{100.0} & 1.0 & gradient & 0.31 & 0.12 & 0.30 & 0.12 \\
     & & \underline{e-geodesic} & 0.78 & 0.06 & 0.77 & 0.06 \\
     & & m-geodesic & 0.17 & 0.10 & 0.17 & 0.11 \\ \hline
    \multirow{3}{*}{100.0} & 100.0 & gradient & 0.31 & 0.12 & 0.30 & 0.12 \\
     & & e-geodesic & 0.71 & 0.14 & 0.70 & 0.14 \\
     & & m-geodesic & 0.18 & 0.10 & 0.17 & 0.11 \\ \hline
    \end{tabular}
\end{table}
From Tables~\ref{fo_table} and \ref{st_table},
we observe that the e-geodesic update with step size $lr=1.0$
consistently achieves the highest accuracy among all methods.

Increasing the number $K$ of random samples for Monte Carlo approximation
results in higher computational cost,
but it is expected to make the approximation in the proposed method more accurate.

On the other hand, the accuracy of the m-geodesic method is relatively low.
This is because the domain of the m-affine coordinates in the multivariate diagonal Gaussian family is restricted by parameter constraints, namely
$\{\eta \in \mathbb{R}^{2n} \mid \eta_i^2 < \eta_{n+i} \ (i = 1, \dots, n)\}$,
whereas the e-affine coordinates can take any value in $\mathbb{R}^n \times {\mathbb{R}_{<0}}^n$.
As a result, while the e-geodesic update can often use the initial step size $lr$ without modification,
the m-geodesic update requires repeated halving of the step size,
i.e.\ $lr \leftarrow lr/2$, in order to remain in the parameter domain.
Thus, only a very small update can be taken in a single iteration,
which explains the inferior performance of the m-geodesic method.

\end{example}

\section{Conclusion}

We introduced geodesic-based optimization methods on dually flat spaces,
which are fundamental geometric structures in information geometry.
The most important application of this framework is maximum likelihood estimation
on exponential families, where the geometry is naturally induced
by the e- and m-connections.

We developed optimization algorithms that update parameters along e- and m-geodesics
in the statistical manifold.
For exponential families, this enables efficient movement toward the maximum likelihood estimator
while respecting the intrinsic geometric structure of the parameter space.
Furthermore, for mixture families, which are also dually flat,
we demonstrated that the proposed geodesic methods outperform the exponentiated gradient method
in terms of the number of iterations required for convergence.

We also applied the proposed method to variational inference
for multinomial logistic regression.
The e-geodesic update with step size $1$ achieved highly accurate estimates
in a single iteration, improving predictive performance
compared with conventional gradient-based optimization.
The effectiveness of the proposed approach is attributed to the relationship
between the objective function and the geometric structure of the statistical manifold.
For example, maximum a posteriori estimation or regularized log-likelihood maximization
still retain geometric alignment between the log-likelihood term
and the manifold structure, and therefore are expected to benefit from
geodesic optimization as well.

\appendix
\section*{Appendix}

\section{Data generation procedure for $X$ and $Y$ in Example \ref{MLR}}
The dataset $(X,Y)$ used in Example \ref{MLR} was generated by the
\texttt{make\_classification} function provided in the Python library
\texttt{scikit-learn}.
The internal mechanism of this function can be described as follows.

The index set $\{1,\dots,N\}$ is partitioned into $D$ subsets
$S_1,\dots,S_D$ in order, so that each subset has either
$\lceil N/D \rceil$ or $\lfloor N/D \rfloor$ elements.
Let $r = N \bmod D$, and $\lceil N/D \rceil$ elements are assigned to
$S_1,\dots,S_r$ and $\lfloor N/D \rfloor$ elements to
$S_{r+1},\dots,S_D$.

The standard basis vectors $e_1,\dots,e_D \in \mathbb{R}^D$ are used,
and $y_i$ is set to $e_j$ for $i \in S_j$, so that each $y_i$
represents its cluster label as a one-hot vector.

Then, $D$ distinct vertices from the hypercube $[-1.5,1.5]^M$ are chosen
as the cluster centers $c_1,\dots,c_D \in \mathbb{R}^M$.
For each $j$, a random matrix $A_j \in \mathbb{R}^{M\times M}$ is generated,
whose entries are independently drawn from the uniform distribution on $[-1,1)$.
Each sample is drawn independently as
$ x_i \sim \mathcal{N}(c_j, A_j^\top A_j)
\qquad (i \in S_j)$.

With probability $0.03$, $y_i$ is replaced by a randomly selected one-hot
vector from $\{e_1,\dots,e_D\}$, which introduces label noise.

Finally, the indices of all pairs $(x_i, y_i)$ are randomly shuffled,
and the feature coordinates of each $x_i$ are also randomly permuted.

\section*{Acknowledgments}
This work was supported by JSPS KAKENHI Grant Numbers JP22H00510 and 25H01156.


\begin{thebibliography}{}
%
\bibitem{amari2000} Amari, S. and Nagaoka, H. (2000).
\textit{Methods of Information Geometry}. American Mathematical Society.
%
\bibitem{BN1978}
Barndorff-Nielsen, O. (1978).
\textit{Information and Exponential Families in Statistical Theory}. Wiley.
%
\bibitem{beck2017}
Beck, A. (2017).
\textit{First-Order Methods in Optimization}. SIAM.
%
\bibitem{blundell2015} Blundell, C., Cornebise, J., Kavukcuoglu, K., and Wierstra, D. (2015).
Weight Uncertainty in Neural Networks.
In \textit{Proceedings of the 32nd International Conference on Machine Learning},
pp.~1613--1622.
%
\bibitem{BA1952} Bradley, R. A. and Terry, M. E. (1952).
Rank analysis of incomplete block designs.
\textit{Biometrika}, 39, pp.~324--345.
%
\bibitem{fisher1925}
Fisher, R. A. (1925).
Theory of Statistical Estimation.
\textit{Proceedings of the Cambridge Philosophical Society}, 22, pp.~700--715.
%
\bibitem{GM1997}
Gourieroux, C. and Monfort, A. (1997).
\textit{Simulation-Based Econometric Methods}. Oxford University Press.
%
\bibitem{kingma2014} Kingma, D. P. and Welling, M. (2014).
Auto-Encoding Variational Bayes.
In \textit{Proceedings of the 2nd International Conference on Learning Representations}.
%
\bibitem{kivinen1997} Kivinen, J. and Warmuth, M. K. (1997).
Exponentiated Gradient versus Gradient Descent for Linear Predictors.
\textit{Information and Computation}, 132, pp.~1--63.
%
\bibitem{lange2013}
Lange, K. (2013).
\textit{Optimization}, 2nd ed. Springer.
%
\bibitem{raskutti2015} Raskutti, G. and Mukherjee, S. (2015).
The Information Geometry of Mirror Descent.
\textit{IEEE Transactions on Information Theory}, 61, pp.~1451--1457.
%
\bibitem{rezende2014} Rezende, D. J., Mohamed, S., and Wierstra, D. (2014).
Stochastic Backpropagation and Approximate Inference in Deep Generative Models.
In \textit{Proceedings of the 31st International Conference on Machine Learning},
pp.~1278--1286.
%
\bibitem{smith1994} Smith, S. T. (1994).
Optimization Techniques on Riemannian Manifolds.
\textit{Fields Institute Communications}, 3, pp.~113--135.
%
\bibitem{suyama2017} Suyama, A. (2017).
\textit{Bayesian Inference in Machine Learning}. Kodansha. (in Japanese)
%
\end{thebibliography}
\end{document}